\documentclass{scrartcl}

\usepackage{cmap}
\usepackage[utf8]{inputenc}
\usepackage[T1]{fontenc}

\usepackage{amsmath}
\usepackage{amssymb}
\usepackage{amsthm}
\usepackage{thmtools}
\usepackage{bm}
\usepackage{mathrsfs}
\usepackage{csquotes}
\usepackage{tabularx}
\usepackage{enumitem}
\usepackage{hyperref}
\usepackage{algorithm}
\usepackage{algpseudocode}
\usepackage{bbm}
\usepackage{textcomp}
\usepackage{varwidth}
\usepackage{tikz}
\usetikzlibrary{automata, calc, matrix, shapes, positioning, decorations.pathreplacing, decorations.pathmorphing, fit}

\usepackage{algorithm}
\usepackage{listings}
\usepackage[lighttt]{lmodern}

\lstnewenvironment{pseudocode}[1][]
{
  \lstset{
    mathescape=true,
    escapeinside={(*@}{@*)},
    numbers=left,
    numberstyle=\tiny,
    basicstyle=\scriptsize\ttfamily,
    keywordstyle=\ttfamily\bfseries,
    keywords={,var, const, begin, end, return, procedure, function, if, then, else, fi, for, each, while, do, od, true, false, not, break, accept, reject, fail,}
    numbers=left,
    xleftmargin=.04\textwidth,
    #1
  }
}
{}

\usepackage{todonotes}

\setcapindent{0pt}

\theoremstyle{plain}
\newtheorem{thm}{Theorem}[section]
\newtheorem{lemma}[thm]{Lemma}
\newtheorem{prop}[thm]{Proposition}
\newtheorem{cor}[thm]{Corollary}
\newtheorem{fact}[thm]{Fact}

\theoremstyle{definition}

\newtheorem{prob}[thm]{Open Problem}

\theoremstyle{remark}
\newtheorem{remark}[thm]{Remark}

\newcommand{\inverse}[1]{\mkern 1.5mu\overline{\mkern-1.5mu#1\mkern-1.5mu}\mkern 1.5mu}

\newsavebox{\blankboxdisplay}
\savebox{\blankboxdisplay}{\hspace{0.1ex}\tikz[baseline=0.1em]{%
    \node [shape=rectangle, anchor=south, draw, solid, inner sep=0pt, minimum width=1ex, minimum height=0.9em] (char) {};}%
  \hspace{0.1ex}}
\newsavebox{\blankboxtext}
\savebox{\blankboxtext}{\hspace{0.1ex}\tikz[baseline=0.1em]{%
    \node [shape=rectangle, anchor=south, draw, solid, inner sep=0pt, minimum width=1ex, minimum height=0.9em] (char) {};}%
  \hspace{0.1ex}}
\newsavebox{\blankboxscript}
\savebox{\blankboxscript}{\scriptsize\hspace{0.1ex}\tikz[baseline=0.1em]{%
    \node [shape=rectangle, anchor=south, draw, solid, inner sep=0pt, minimum width=1ex, minimum height=0.9em] (char) {};}%
  \hspace{0.1ex}}
\newsavebox{\blankboxscriptscript}
\savebox{\blankboxscriptscript}{\tiny\hspace{0.1ex}\tikz[baseline=0.1em]{%
    \node [shape=rectangle, anchor=south, draw, solid, inner sep=0pt, minimum width=1ex, minimum height=0.9em] (char) {};}%
  \hspace{0.1ex}}
\newcommand{\blank}{\makeatletter%
  \ensuremath{\mathchoice%
    {\text{\usebox\blankboxdisplay}}%
    {\text{\usebox\blankboxtext}}%
    {\text{\usebox\blankboxscript}}%
    {\text{\usebox\blankboxscriptscript}}}%
  \makeatother}

\newlength{\edgelength}
\newcommand{\trans}[4]{%
  \begin{tikzpicture}[auto, shorten >=1pt, >=latex, baseline=(l.base), inner sep=0pt, outer xsep=0.3333em]
    \node (l) {\ensuremath{#1}};%
    \setlength{\edgelength}{\widthof{\scriptsize\ensuremath{#2/#3}}+0.5cm}%
    \node[base right=\edgelength of l] (r) {\ensuremath{#4}};%
    \path[->] (l.mid east) edge node[inner sep=0pt] {\scriptsize\ensuremath{#2/#3}} (r.mid west);%
  \end{tikzpicture}%
}

\newcommand{\transa}[3]{%
  \begin{tikzpicture}[auto, shorten >=1pt, >=latex, baseline=(l.base), inner sep=0pt, outer xsep=0.3333em]
    \node (l) {\ensuremath{#1}};%
    \setlength{\edgelength}{\widthof{\scriptsize\ensuremath{#2}}+0.5cm}%
    \node[base right=\edgelength of l] (r) {\ensuremath{#3}};%
    \path[->] (l.mid east) edge node[inner xsep=0pt, inner ysep=0.2em] {\scriptsize\ensuremath{#2}} (r.mid west);%
  \end{tikzpicture}%
}

\DeclareMathOperator{\id}{id}
\DeclareMathOperator{\Stab}{Stab}

\newcommand{\wt}{\widetilde}
\DeclareMathOperator{\idGrp}{1}

\newcommand{\problem}[3][]{%
  \par\vspace{0.125cm plus 0.05cm minus 0.05cm}\begin{tabularx}{\textwidth-2\parindent}{lX}%
    \if\relax\detokenize{#1}\relax%
    \else%
      \textnormal{\textbf{Constant:}}&#1\\%
    \fi%
    \textnormal{\textbf{Input:}}&#2\\%
    \textnormal{\textbf{Question:}}&#3\\%
  \end{tabularx}\vspace{0.125cm plus 0.05cm minus 0.05cm}\par%
  }

\usepackage[affil-it]{authblk}

\author{Daniele D'Angeli\thanks{The first author was supported by the Austrian Science Fund project FWF P29355-N35.}}
\affil{Università degli Studi Niccolò Cusano\\
  Via Don Carlo Gnocchi, 3\\
  00166 Roma, Italia}
\author{Emanuele Rodaro}
\affil{Department of Mathematics\\
  Politecnico di Milano\\
  Piazza Leonardo da Vinci, 32\\
  20133 Milano, Italy}
\author{Jan Philipp Wächter}
\affil{Institut für Formale Methoden der Informatik (FMI)\\
  Universität Stuttgart\\
  Universitätsstraße 38\\
  70569 Stuttgart, Germany}

\title{Orbit Expandability of Automaton Semigroups and Groups}

\begin{document}
  \maketitle

  \vspace{-2\baselineskip}
  \begin{abstract}
    We introduce the notion of \emph{expandability} in the context of automaton semigroups and groups: a word is \emph{$k$-expandable} if one can append a suffix to it such that the size of the orbit under the action of the automaton increases by at least $k$. This definition is motivated by the question which $\omega$-words admit infinite orbits: for such a word, every prefix is expandable.
    
    In this paper, we show that, given a word $u$, an automaton $\mathcal{T}$ and a number $k$ as input, it is decidable to check whether $u$ is $k$-expandable with respect to the action of $\mathcal{T}$. In fact, this can be done in exponential nondeterministic space. From this nondeterministic algorithm, we obtain a bound on the length of a potential orbit-increasing suffix $x$. Moreover, we investigate the situation if the automaton is invertible and generates a group. In this case, we give an algebraic characterization for the expandability of a word based on its \emph{shifted stabilizer}. We also give a more efficient algorithm to decide the expandability of a word in the case of automaton groups, which allows us to improve the upper bound on the maximal orbit-increasing suffix length. Then, we investigate the situation for reversible (and complete) automata and obtain that every word is expandable with respect to such automata. Finally, we give a lower bound example for the length of an orbit-increasing suffix.
    
  \end{abstract}

  \begin{section}{Introduction}
    The class of groups generated by finite-state, letter-to-letter transducers -- the, so called, \emph{automaton groups} -- is famous as a rich source for examples of groups with interesting or even astonishing properties. Probably, the best known among these examples is Grigorchuk's group\footnote{See \cite{grigorchuk2008groups} for an accessible introduction to Grigorchuk's group.}, which is the historically first example of a group with sub-exponential but super-polynomial growth. Growth, in this context, refers to the growth of the function mapping a natural number $n$ to the number of elements reachable from the identity in the Cayley graph of the group in at most $n$ steps. With its intermediate growth, Grigorchuk's group answered the famous question by Milnor on the existence of such groups. This answer not only stirred interest in Grigorchuk's group but also in the overall class of automaton groups.
    
    Part of the study of automaton groups is to investigate their Schreier graphs. Since automaton groups act naturally on the words over some finite alphabet, there is a connection between Schreier graphs and orbital graphs of right-infinite words: the orbital graph of $\xi$ is isomorphic to the Schreier graph for the stabilizer of $\xi$. In fact, if the stabilizer is trivial, then the Schreier graph coincides with the Cayley graph. With this in mind, one motivation for the interest in Schreier or orbital graphs is that it is possible to generalize the notion of group growth\footnote{For both of these notions of growth, we refer the reader to Grigorchuk's survey \cite{grigorchuk2011milnor} for an introduction.}: how many elements are reachable in the orbital graph of $\xi$ in at most $n$ steps starting from $\xi$? In fact, considering orbital graphs instead of Schreier graphs allows to also consider the generalized setting of automaton \emph{semi}groups.
    
    In this paper, we are interested in a somewhat dual notion of orbital growth: if we start with a finite word $u$, how much longer do we have to make $u$ in order to obtain a larger orbit or what do we have to append to $u$ for this to happen? Formally, a finite word $u$ is said to be \emph{$k$-expandable} if there is some (finite) word $x$ such that the size of the orbit of $ux$ is larger than the size of the orbit of $u$ by at least $k$. Of course, behind this definition stands a more fundamental question: which (infinite) words admit an infinite orbit under the action of an automaton? Clearly, all prefixes of a word with an infinite orbit must be expandable! In fact, they are $k$-expandable for arbitrary $k$. On the other hand, if a word $u$ is $k$-expandable for some $k$, then there is some suffix $x$ such that $ux$ has a larger orbit. If we can choose $x$ in such a way that $ux$ remains expandable itself, then this yields a word with infinite orbit in the limit.
    
    The study of expandability is motivated by a question asked by the authors in \cite[Open Problem~4.3]{decidabilityPart}: does every infinite automaton semigroup admit a word with infinite orbit? In fact, it turns out that the answer to this question is \enquote{yes} \cite{orbitsPart}. However, the proof for this result is purely existential and does not give information on the structure of words with infinite orbit. On the other hand, it shows that deciding the finiteness problem for automaton semigroups and groups is equivalent to deciding whether there is a word with infinite orbit. While Gillibert proved that the finiteness problem is undecidable in the semigroup case \cite{Gilbert13}, it remains an important open problem in the group case.
    
    In the light of this open problem, results on the structure of words with infinite orbits seem highly interesting and the goal of this paper is to contribute to this study. Our discussion is centered around some of the most obvious questions related to the notion of expandability. Is it decidable to check whether a word is expandable? Can we characterize expandable words? Can we say something about the suffix which needs to be appended in order to increase the orbit size? Is the situation different with complete, invertible or reversible automata?
    
    Before giving answers to some of these questions, we introduce basic definitions in \autoref{sec:preliminaries}. Then, in \autoref{sec:decidable}, we state and analyze a nondeterministic space-bounded algorithm for the more general semigroup case to decide whether a given word $u$ is expandable. The nature of the algorithm will also allow us to give an upper bound on $|x|$ where $x$ is a suffix of minimal length such that the orbit of $ux$ is larger by at least $k$ than the orbit of $u$. In \autoref{sec:inverse}, we use the same approach to obtain a better upper bound in the case of automaton groups. First, we give an algebraic characterization of expandable words based on their stabilizer. Then, we use this characterization to obtain a more efficient nondeterministic algorithm yielding the better upper bound. Afterwards, in \autoref{sec:reversible}, we consider complete and reversible automata and see that every word is expandable with respect to these automata. Similarly to before, we obtain an even better upper bound on the length of the suffix. We conclude the paper with a lower bound construction for the length on the witness word in the group case in \autoref{sec:lowerBound}.
  \end{section}

  \begin{section}{Preliminaries}\label{sec:preliminaries}
    \paragraph{Words, Functions and Algebra.}
    Let $\Sigma$ be an \emph{alphabet}, i.\,e., a non-empty, finite set. Then, we use $\Sigma^*$ to denote the set of (finite) words over $\Sigma$ including the empty word, which we denote by $\varepsilon$. To denote the set of words excluding the empty one, we use $\Sigma^+$. Additionally, we use natural notations such as $\Sigma^n$ for the set of words $u$ with \emph{length} $|u| = n$ and $\Sigma^{< n}$ for the words of length smaller than $n$. A word $u$ is a \emph{prefix} of some other word $v$ if there is a word $x$ such that $v = ux$.
    
    To denote a \emph{partial} function $f$ from a set $A$ to a set $B$, we use the notation $f \colon A \to_p B$ (and omit the index $p$ if the function is total). For partial functions, we use the term one-to-one instead of injective: $f$ is \emph{one-to-one} if $f(a) = f(a')$ implies $a = a'$. For the disjoint union of two sets $A$ and $B$, we write $A \sqcup B$.
    
    We assume the reader to be familiar with basic algebraic notions such as semigroups and groups. We want to point out the difference between semigroup inverses and (group) inverses, however. An element $\inverse{s}$ of a semigroup $S$ is \emph{semigroup inverse} to another element $s \in S$ if $s \inverse{s} s = s$ and $\inverse{s} s \inverse{s} = \inverse{s}$ hold. Inverses in the usual group sense are, in particular, semigroup inverses but the converse is not true in general. A semigroup is an \emph{inverse} semigroup if every element has a \emph{unique} inverse.
    
    As a method to graphically depict the multiplication on a semigroup $S$ generated by a set $A$, we use the concept of a left Cayley Graph: its nodes are the elements of $S$ and it contains the edges $\transa{s}{a}{as}$ for every $s \in S$ and $a \in A$. Additionally, we use the \emph{eggbox} presentation for finite semigroups (see, e.\,g., \cite[p.~48]{howie}).
    
    \paragraph*{Decidability, $\mathcal{O}$-notation and Nondeterminism.}
    We only need basic ideas from the theory of decidability. However, we assume the reader to be somewhat familiar with $\mathcal{O}$-notation as we will make lax use of it. Some of our algorithms will be nondeterministic space-bounded algorithms, so we will need the concept of nondeterministic computations. As we often make use of \enquote{guess and check} algorithms, some familiarity with the concept will be beneficial. The only advanced method from complexity theory used in the paper is the inductive counting technique (famously applied to show the closure of nondeterministic space classes under complementation, the Immerman–Szelepcsényi theorem \cite{immerman1988nondeterministic, szelepcsenyi1988method}; see \cite[Theorem~7.6]{Pap94} for the theorem and the technique). However, understanding of this technique is only required to obtain a stronger space bound which is not vital for the other parts of the paper.
    
    \paragraph{Automata.}
    As it is common in the field of automaton semigroups and groups, we use the term automaton for what is more precisely referred to as a finite-state, letter-to-letter transducer without initial or final states whose input and output alphabets coincide. Formally, an \emph{automaton} is a triple $\mathcal{T} = (Q, \Sigma, \delta)$ where $Q$ is a finite set of \emph{states}, $\Sigma$ is an alphabet and $\delta \subseteq Q \times \Sigma \times \Sigma \times Q$ is a set of \emph{transitions}. As it is clearer, we use the notation $\trans{q}{a}{b}{p}$ instead of $(q, a, b, p)$ for transitions. Additionally, we use the common way of depicting automata:
    
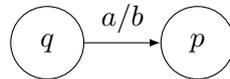
\begin{center}
      \begin{tikzpicture}[>=latex, auto,]
        \node[state] (q) {$q$};
        \node[state, right=of q] (p) {$p$};
        \draw[->] (q) edge node {$a / b$} (p);
      \end{tikzpicture}
    \end{center}
    indicates that the automaton contains states $q$ and $p$ and a transition $\trans{q}{a}{b}{p}$.
    
    An automaton $\mathcal{T} = (Q, \Sigma, \delta)$ is called \emph{deterministic} if
    \[
      \left| \left\{ \trans{q}{a}{b}{p} \in \delta \mid p \in Q, b \in \Sigma \right\} \right| \leq 1
    \]
    holds for every $q \in Q$ and every $a \in \Sigma$. Similarly, it is \emph{complete} if
    \[
      \left| \left\{ \trans{q}{a}{b}{p} \in \delta \mid p \in Q, b \in \Sigma \right\} \right| \geq 1
    \]
    holds. Finally, it is called \emph{reversible} if it is co-deterministic with respect to the input, i.\,e., every state $q \in Q$ has at most one in-going transition with input label $a$ for every $a \in \Sigma$:
    \[
      \left| \left\{ \trans{p}{a}{b}{q} \in \delta \mid p \in Q, b \in \Sigma \right\} \right| \leq 1 \text{.}
    \]
    
    A \emph{run} of the automaton $\mathcal{T} = (Q, \Sigma, \delta)$ is a sequence
    \begin{center}
      \begin{tikzpicture}[auto, >=latex]
        \node (q0) {$q_0$};
        \node[right=of q0] (q1) {$q_1$};
        \node[right=of q1] (dots) {$\dots$};
        \node[right=of dots] (qn) {$q_n$};
        \draw[->] (q0) edge node {\scriptsize$a_1 / b_1$} (q1)
                  (q1) edge node {\scriptsize$a_2 / b_2$} (dots)
                  (dots) edge node {\scriptsize$a_n / b_n$} (qn);
      \end{tikzpicture}
    \end{center}
    such that $\trans{q_{i - 1}}{a_i}{b_i}{q_i}$ is a transition in $\delta$ for every $1 \leq i \leq n$. The run is said to \emph{start} in $q_0$ and \emph{end} in $q_n$. Its input is $a_1 \dots a_n$ and its output is $b_1 \dots b_n$.
    
    \paragraph{Automaton Semigroups.}
    In a deterministic automaton $\mathcal{T} = (Q, \Sigma, \delta)$, every state $q \in Q$ induces a partial function $q \circ{}\! \colon \Sigma^* \to_p \Sigma^*$. This function maps a word $u = a_1 \dots a_n$ with $a_1, \dots, a_n \in \Sigma^*$ to the word $v = b_1 \dots b_n$ if the automaton admits a run starting in $q$ with input $u$ and output $v$. Although such a run does not necessarily exist (because the automaton may not be complete), it is unique if it does (this is due to determinism of the automaton), so the partial function is well-defined. Also note that the image of a word of length $n$ under $q \circ{}\!$ must be a word of length $n$ itself; all $q \circ{}\!$ are \emph{length-preserving}. We simply write $q_\ell \dots q_1 \circ{}\!$ for the composition $q_\ell \circ \dots q_1 \circ{}\!$ of the partial functions $q_1 \circ{}\!, \dots, q_\ell \circ{}\!$ with $q_1, \dots, q_\ell \in Q$. With this notation, the closure $\{ \bm{q} \circ{}\! \mid \bm{q} \in Q^+ \}$ of the functions induced by the states under composition of partial functions is a semigroup: the semigroup \emph{generated} by the automaton $\mathcal{T}$. To emphasize the fact that they generate semigroups, we use the name \emph{$S$-automaton} for a deterministic automaton from now on. A semigroup which is generated by some $S$-automaton is called an \emph{automaton semigroup}.
    
    Note that the functions $q \circ{}\!$ with $q \in Q$ and, more generally, the functions $\bm{q} \circ{}\!$ with $\bm{q} \in Q^+$ are total if the automaton is complete -- which we do not require in general!
    
    \begin{remark}
      This is different to many other papers on automaton semigroups (it does not make any difference for automaton groups as we will see shortly): they define automaton semigroups to be generated by \emph{complete}, deterministic automata. The advantage of defining automaton semigroups using partial automata is that this definition allows for a natural presentation of inverse semigroups as automaton semigroups. Additionally, since we mostly consider algorithms for automaton semigroup in this paper, it seems adequate to handle the more general case of semigroups generated by possibly non-complete automata.
      
      In fact, it is an open problem whether the class of semigroups generated by (partial/\nolinebreak{}possibly non-complete) $S$-automata coincides with the class of semigroups generated by complete $S$-automata. On the other hand, it is not too difficult to see that $S^0$, the semigroup arising from $S$ by adjoining a zero\footnote{A \emph{zero} of a semigroup $S$ is an element $0 \in S$ with $0 s = s 0 = 0$ for all $s \in S$.}, is generated by a complete $S$-automaton if $S$ is generated by a (possibly non-complete) $S$-automaton. More information on the connection between the two notions can be found in \cite{structurePart}.
    \end{remark}
    
    \paragraph{Dual Action.}
    The partial functions $\bm{q} \circ{}\! \colon \Sigma^* \to_p \Sigma^*$ belonging to an $S$-automaton $\mathcal{T} = (Q, \Sigma, \delta)$ represent a (partial) action of $Q^*$ on $\Sigma^*$. However, we also have an action of $\Sigma^*$ on $Q^*$, for which we define dual partial functions $\!{}\cdot u \colon Q^* \to_p Q^*$ for $u \in \Sigma^*$. On the states, the partial function $\!{}\cdot u$ maps $q \in Q$ to $q \cdot u = p \in Q$ if there is a run with input $u$ starting in $q$ and ending in $p$ in the automaton $\mathcal{T}$. Since the run must be unique if it exists, $\!{}\cdot u$ is well-defined on the states. To extend $\!{}\cdot u$ into a partial function $Q^* \to_p Q^*$, we define $\varepsilon \cdot u = \varepsilon$ and, inductively, $\bm{p} q \cdot u = \left(\bm{p} \cdot (q \circ u) \right) (q \cdot u)$ for $\bm{p} \in Q^+$ and $q \in Q$. An easy calculation shows $\bm{q} \cdot u \cdot v = \bm{q} \cdot uv$ for all $\bm{q} \in Q^*$ and $u, v \in \Sigma^*$.
    
    Some properties of the $S$-automaton $\mathcal{T}$ are reflected in the partial functions ${}\! \cdot u$. For example, similar to the partial functions $\bm{q} \circ{}\!$, we have that $\!{} \cdot u$ is total if the automaton is complete. If $\mathcal{T}$ is reversible, then the partial functions $\!{} \cdot u$ are one-to-one:
    \begin{fact}\label{fct:reversibleImpliesOneToOne}
      Let $\mathcal{T} = (Q, \Sigma, \delta)$ be a reversible $S$-automaton, then all partial functions $\!{}\cdot u$ with $u \in \Sigma^*$ are one-to-one.
    \end{fact}
    \begin{proof}
      We have to show that $\bm{p} \cdot u = \bm{q} \cdot u$ implies $\bm{p} = \bm{q}$ whenever $\!{}\cdot u$ is defined on $\bm{p} \in Q^+$ and $\bm{q} \in Q^+$ (since the case $\varepsilon \cdot u$ is trivial). We do this by induction on $|\bm{p}| = |\bm{q}|$. For a single state $q \in Q$ and a single letter $a \in \Sigma$, we have that $q$ has at most one in-going transition with input $a$ by the definition of reversibility. Iterating this argument, we also see that there is at most one run with input $u$ ending in $q$. This shows $p \cdot u = q \cdot u \implies p = q$ for $p, q \in Q$. If we have $\bm{p} p \cdot u = \bm{q} q \cdot u$ for $\bm{p}, \bm{q} \in Q^+$ and $p, q \in Q$, we have
      \[
        \left( \bm{p} \cdot (p \circ u) \right) (p \cdot u) = \bm{p} p \cdot u = \bm{q} q \cdot u = \left( \bm{q} \cdot (q \circ u) \right) (q \cdot u) \text{,}
      \]
      which implies $\bm{p} \cdot (p \circ u) = \bm{q} \cdot (q \circ u)$ and $p \cdot u = q \cdot u$. By induction, we get $\bm{p} = \bm{q}$ and also $p = q$.
    \end{proof}

    \paragraph{Automaton Constructions.}
    For two automata $\mathcal{T}_1 = (Q_1, \Sigma_1, \delta_1)$ and $\mathcal{T}_2 = (Q_2, \Sigma_2, \delta_2)$, we can make their \emph{disjoint union} $\mathcal{T}_1 \sqcup \mathcal{T}_2 = (Q_1 \sqcup Q_2, \Sigma_1 \cup \Sigma_2, \delta_1 \sqcup \delta_2)$. An important observation here is that the disjoint union of two $S$-automata is an $S$-automaton again; similarly, the disjoint union of two complete automata over the same alphabet remains complete.
    
    More interesting than the disjoint union of automata is the construction of inverse automata. Let $\mathcal{T} = (Q, \Sigma, \delta)$ be an automaton. Then, its \emph{inverse} automaton is $\inverse{\mathcal{T}} = (\inverse{Q}, \Sigma, \inverse{\delta})$ where $\inverse{Q}$ is a disjoint copy of $Q$ and the transitions are given by
    \[
      \inverse{\delta} = \{ \trans{q}{b}{a}{p} \mid \trans{q}{a}{b}{p} \in \delta \} \text{.}
    \]
    In general, $\inverse{\mathcal{T}}$ will not be deterministic (even if $\mathcal{T}$ is). However, if it is deterministic, we call $\mathcal{T}$ \emph{inverse-deterministic} or \emph{invertible}. In this case, $\inverse{\mathcal{T}}$ is an $S$-automaton and every state $\inverse{q} \in \inverse{Q}$ induces a function $\inverse{q} \circ{}\!$. By construction, it is easy to see that $\inverse{q} \circ{}\!$ is inverse in the sense of a semigroup inverse to the function $q \circ{}\!$ if $\mathcal{T}$ is not only inverse-deterministic but also deterministic. Motivated by this observation, we define $\inverse{\inverse{q}} = q$ and $\inverse{\bm{q}} = \inverse{q_1} \dots \inverse{q_n}$ for $\bm{q} = q_n \dots q_1$ with $q_1, \dots, q_n \in Q$. This turns taking the inverse into an involution as we have $\inverse{\inverse{\mathcal{T}}} = \mathcal{T}$.

    \paragraph{Inverse Automaton Semigroups and Automaton Groups.}
    It is not difficult to see that the functions $q \circ{}\!$ induced by the states $q \in Q$ of an inverse-deterministic $S$-automaton $\mathcal{T} = (Q, \Sigma, \delta)$ are partial one-to-one functions and that the semigroup $\mathscr{S}(\mathcal{T} \sqcup \inverse{\mathcal{T}})$ generated by the union of $\mathcal{T}$ and its inverse $\inverse{\mathcal{T}}$ is an inverse semigroup. We define $\inverse{\mathscr{S}}(\mathcal{T}) = \mathscr{S}(\mathcal{T} \sqcup \inverse{\mathcal{T}})$ and call this the \emph{inverse semigroup generated by $\mathcal{T}$}. Similar to the definition of $S$-automata, we call an inverse-deterministic $S$-automaton an \emph{$\inverse{S}$-automaton}. An inverse semigroup which is generated by some $\inverse{S}$-automaton is called an \emph{inverse automaton semigroup}. Please note that, a priori, there is a difference between an inverse automaton semigroup and an automaton semigroup which happens to be inverse. However, the name is justified as the two notions coincide \cite[Theorem~25]{structurePart}.
    
    We have already observed above that the functions $\bm{q} \circ{}\!$ in an automaton semigroup generated by a complete $S$-automaton are total. If the generating automaton is additionally inverse-deterministic, they are also one-to-one or injective. Because they are length-preserving, they can be restricted into (injective) functions $\Sigma^n \to \Sigma^n$ for every $n \geq 0$. Due to cardinality, these restrictions must be surjective and, thus, bijections. This implies that $\bm{q} \circ{}\! \colon \Sigma^* \to \Sigma^*$ must be a bijection as well and that the inverse semigroup generated by such an automaton is a group. To emphasize this fact, we write $\mathscr{G}(\mathcal{T})$ instead of $\inverse{\mathscr{S}}(\mathcal{T})$ in the case of complete $\inverse{S}$-automata and say that $\mathscr{G}(\mathcal{T})$ is the \emph{group generated by $\mathcal{T}$}. Additionally, we call a complete $\inverse{S}$-automaton a \emph{$G$-automaton} from now on. A group generated by some $G$-automaton is an \emph{automaton group}.

    \paragraph{Orbital and Schreier Graphs.}
    For an $S$-automaton $\mathcal{T}$, we define $\mathcal{T} \circ u$ as the \emph{orbital graph} of $u \in \Sigma^*$. Its node set is the orbit $Q^* \circ u = \{ \bm{q} \circ u \mid \bm{q} \in Q^*, \bm{q} \circ{}\! \text{ defined on $u$} \}$ of $u$ and it contains an edge $\transa{v}{q}{w}$ labeled with $q \in Q$ whenever $w = q \circ v$. An important observation to make in this context is that, for a path
    \begin{center}
      \begin{tikzpicture}[>=latex, auto]
        \node (v0) {$v_0$};
        \node[right=of v0] (v1) {$v_1$};
        \node[right=of v1] (dots) {$\dots$};
        \node[right=of dots] (vn) {$v_n$};
        \draw[->] (v0) edge node {$q_1$} (v1)
                  (v1) edge node {$q_2$} (dots)
                  (dots) edge node {$q_n$} (vn);
      \end{tikzpicture}
    \end{center}
    in the orbital graph, we actually have $v_n = q_n \dots q_1 \circ v_0$ and \emph{not} $v_n = q_1 \dots q_n \circ v_0$ (i.\,e., we have to read the label of the path right to left). This is because we let the semigroups act on the left.
    
    If $\mathcal{T}$ is additionally inverse-deterministic, we can also define $\widetilde{\mathcal{T}} \circ u$, the \emph{Schreier graph} of $u \in \Sigma^*$. It is similar to the orbital graph but it also contains edges for the states of the inverse $\inverse{\mathcal{T}}$. Thus, as its state set, we define $\widetilde{Q} \circ u = \{ \widetilde{\bm{q}} \circ u \mid \widetilde{\bm{q}} \in \widetilde{Q}^*, \widetilde{\bm{q}} \circ{}\! \text{ defined on $u$} \}$ where $\widetilde{Q} = Q \sqcup \inverse{Q}$ is the union of the states of $\mathcal{T}$ and $\inverse{\mathcal{T}}$. In addition to the edges $\transa{v}{q}{w}$ with $q \in Q$ if $w = q \circ v$ (as in the orbital graph), we also define the edges $\transa{w}{\inverse{q}}{v}$ for $\inverse{q} \in \inverse{Q}$ if $v = \inverse{q} \circ w$. With this definition, $\widetilde{\mathcal{T}} \circ u$ is a super-graph of $\mathcal{T} \circ u$. If $\mathcal{T}$ is a $G$-automaton, then it turns out that $Q^* \circ u$ and $\widetilde{Q}^* \circ u$ are of the same size and that, so, $\widetilde{\mathcal{T}} \circ u$ is basically the same graph as $\mathcal{T} \circ u$; the only difference is that $\widetilde{\mathcal{T}} \circ u$ additionally contains a back-edge $\transa{w}{\inverse{q}}{v}$ for every edge $\transa{v}{q}{w}$ (see, e.\,g., \cite[Lemma~1]{decidabilityPart} for this result with the same notation as used here). For a non-complete $\inverse{S}$-automaton, the Schreier graph $\widetilde{\mathcal{T}} \circ u$ can be strictly larger than the orbital graph $\mathcal{T} \circ u$, however (see \cite[Lemma~2]{decidabilityPart})!
    
    \paragraph{Expandability.}
    The most central notion for our results is the notion of expandability. Let $\mathcal{T} = (Q, \Sigma, \delta)$ be an $S$-automaton. For a number $k \geq 1$ and a word $x \in \Sigma^*$, we say that a word $u \in \Sigma^{*} $ is \emph{$k$-expandable by $x$} (with respect to $\mathcal{T}$) if $|Q^{*}\circ ux| - |Q^{*} \circ u| \ge k$. The word $u$ is \emph{$k$-expandable} if it is $k$-expandable by some $x \in \Sigma^*$ and it is \emph{expandable} if it is $k$-expandable for some $k \geq 1$. Notice that, in the case of non-complete automata, it is possible that $Q^{*} \circ uy$ is actually smaller than $Q^* \circ u$.
  \end{section}

  \newpage\enlargethispage{2\baselineskip}
  \begin{section}{Expandability is Decidable}\label{sec:decidable}
    We start with a rather simple combinatorial lemma.
    \begin{lemma}\label{lem:expandabilityBounds}
      Let $\mathcal{T} = (Q, \Sigma, \delta)$ be an $S$-automaton and let $u \in \Sigma^*$ be $k$-expandable. Then, there is an $x \in \Sigma^*$ such that
      \[
        n + k \leq \left| Q^* \circ ux \right| < (n + k) \, |\Sigma| \text{,}
      \]
      where $n = \left| Q^* \circ u \right|$.
    \end{lemma}
    \begin{proof}
      Since $u$ is $k$-expandable, there is a $y \in \Sigma^*$ with $n + k \leq \left| Q^* \circ uy \right|$. Let $x'$ denote the longest prefix of $y$ such that $\left| Q^* \circ ux' \right| < n + k$, i.\,e., we have $y = x' a y'$ for some $a \in \Sigma$ and some $y' \in \Sigma^*$. By our choice of $x'$, we have $n + k \leq \left| Q^* \circ ux'a \right|$. As the size of $\left| Q^* \circ wa \right|$ is limited by $\left| Q^* \circ w \right| \left| \Sigma \right|$ for any word $w \in \Sigma^*$, this yields
      \[
        n + k \leq \left| Q^* \circ ux'a \right| \leq \left| Q^* \circ ux' \right| \left| \Sigma \right| < (n + k) \left| \Sigma \right| \text{.}
      \]
      Thus, $x = x'a$ satisfies the inequality in the lemma.
    \end{proof}

    We will use this lemma to prove that it is decidable whether an input word is ($k$-)\allowbreak{}expandable.

    \begin{thm}\label{theo: decidability k-expandability}
      The problem
      \problem{
        an $S$-automaton $\mathcal{T} = (Q, \Sigma, \delta)$,\newline
        a natural number $k$ and a word $u \in \Sigma^*$
      }{
        is $u$ $k$-expandable (with respect to $\mathcal{T}$)?
      }
      \noindent{}can be decided in nondeterministic space $\mathcal{O}((n + k)^2 \, |\Sigma| \, (1 + \log |Q|) + |u| (1 + \log |\Sigma|))$, where $n = \left| Q^* \circ u \right| \leq |\Sigma|^{|u|}$ is the size of the orbit of the input word.
      For the problem
      \problem[
        an $S$-automaton $\mathcal{T} = (Q, \Sigma, \delta)$,
      ]{
        a natural number $k$ and a word $u \in \Sigma^*$
      }{
        is $u$ $k$-expandable (with respect to $\mathcal{T}$)?
      }
      \noindent{}this yields nondeterministic space $2^{\mathcal{O}( |u| + \log k )}$.
    \end{thm}
    \begin{proof}
      First, note that all words in $Q^* \circ u$ are of length $|u|$. Thus, $n$ is bounded by $|\Sigma|^{|u|}$.

      Instead of giving a decision algorithm, we give a (nondeterministic) space-bounded semi-algorithm. The result then follows from the closure of nondeterministic space classes under complement (see, e.\,g., \cite[Theorem~7.6]{Pap94}).

      Note that, on input of $\mathcal{T} = (Q, \Sigma, \delta)$ and $u \in \Sigma^*$, we can compute $n = \left| Q^* \circ u \right|$ using the technique of inductive counting\footnote{Using inductive counting is only necessary to show the space bound stated in the theorem. With a laxer space bound, one could also compute $n$ using a naïve algorithm.} \cite[Theorem~7.6]{Pap94}. In order to do this, we only need to store words in $Q^* \circ u$, a binary counter for $n$ and states in $Q$. As all words in $Q^* \circ u$ have length $|u|$, we can store them in space $\mathcal{O}(|u| (1 + \log |\Sigma|))$. The binary counter for $n$ can also be realized in space $\mathcal{O}(|u| (1 + \log |\Sigma|))$ (since $n \leq |\Sigma|^{|u|}$). Finally, states can be stored in space $\mathcal{O}(1 + \log |Q|)$. Thus, all of these can be stored within the space bound $\mathcal{O}((n + k)^2 \, |\Sigma| \, (1 + \log |Q|) + |u| (1 + \log |\Sigma|))$ mentioned in the theorem.\footnote{Of course, a smaller space bound would suffice here but we need the additional space below.}

      We can solve the main part of the problem using a \enquote{guess and check} approach. We give a rather informal description of the algorithm here; pseudo-code for it can be found in \autoref{alg:expandability}. First, we guess $n + k$ state sequences $\bm{q}_1, \dots, \bm{q}_{n + k}$ of length smaller than $K = (n + k) |\Sigma| - 1$. The idea is that these state sequences lead to different elements in $Q^* \circ ux$ for some witness $x \in \Sigma^*$ for the $k$-expandability of $u$. Next, we compute $\bm{q}_i \cdot u$ for each of the state sequences. Finally, we guess $x$ letter by letter (without storing the previous letters). After we guessed a new letter $b \in \Sigma$, we update the stored state sequences $\bm{p}_i$ to $\bm{p}_i \cdot b$. While we do all of this, we also keep track of the pairs of state sequences for which we have already encountered a difference in their outputs on $u$ followed by the guessed letters. Whenever a transition is not defined, we simply cancel the respective computational branch.

      It is clear that, if the algorithm returns \enquote{$u$ is $k$-expandable}, then this is correct as the guessed state sequences and letters witness the expandability. On the other hand, it is sufficient to only considering state sequences of length smaller than $K$ to discover a witness if one exists. To see this, suppose that $u$ is $k$-expandable. Then, by \autoref{lem:expandabilityBounds}, there is an $x \in \Sigma^*$ such that $u$ is $k$-expandable by $x$ and the orbital graph $\mathcal{T} \circ ux$ is of size $\left| Q^* \circ ux \right| \leq K$. Now, any node in this graph is reachable from $ux$ by a path of length smaller than $K$ as longer paths need to include a loop by the pigeon hole principle. The labelings on these paths correspond to the elements $\bm{q}_i$ and, thus, bounding their length to be smaller than $K$ is no restriction.

      The most interesting part of the space analysis are the variables $\bm{q}_1, \dots, \bm{q}_{n + k}$ and the variable \texttt{differences}. Other variables, like elements of $\Sigma$ and counters up to $|u|$, $n$, $n + k$ or $K$, can certainly be realized in $\mathcal{O}(1 + \log |\Sigma| + \log |u| + \log (n + k))$ and, thus, in the space bound stated in the theorem. For storing a single variable $\bm{q}_i \in Q^{< K}$, we need space smaller than $K (1 + \log |Q|) < (n + k) \, |\Sigma| \, (1 + \log |Q|)$. Thus, for all variables $\bm{q}_1, \dots, \bm{q}_{n + k}$, we need less than $(n + k)^2 \, |\Sigma| \, (1 + \log |Q|)$ space, which is smaller than the required space bound. Finally, for \texttt{differences}, we need to store a bit for the $\binom{n + k}{2}$ many subsets of size $2$ of $\{ 1, \dots, n + k \}$. This yields a space requirement in $\mathcal{O}\left( \binom{n + k}{2} \right) \subseteq \mathcal{O}((n + k)^2)$.
    \end{proof}

    The algorithm also yields an upper bound on the length of a shortest word $x$ witnessing the $k$-expandability of $u$. This upper bound is obtained by counting the possible configurations of the automaton during during the letter guessing phase.
    \begin{algorithm}[!b]
      \begin{pseudocode}
function IskExpandable($\mathcal{T} = (Q, \Sigma, \delta)$, $k$, $u = a_1 \dots a_m$): $\mathbb{B}$;
const
  $n = \left| Q^* \circ u \right|$;
  $K = (n + k) |\Sigma| - 1$;
var
  $\bm{q}_1, \dots, \bm{q}_{n + k} \in Q^{< K}$;
  differences${}\subseteq \left\{ \{ i, j \} \mid i \neq j, 1 \leq i, j \leq n + k \right\}$;(*@ \hfill \normalfont{$\triangleright$ For which pairs have we seen a difference?} @*)
  $b \in \Sigma$;
begin
  for $i \in \{ 1, \dots, n + k \}$ do(*@ \hfill \normalfont{$\triangleright$ Guess initial $\bm{q}^{(0)}_1, \dots, \bm{q}^{(0)}_{n + k}$ values of $\bm{q}_1, \dots, \bm{q}_{n + k}$} @*)
    $\bm{q}_i \gets{}$guess($Q^{< K}$);
  od;
  differences${}\gets \emptyset$;(*@ \hfill \normalfont{$\triangleright$ Compute $\bm{q}^{(0)}_i \cdot u$} while checking for differences between pairs @*)
  for $\ell \in \{ 1, \dots, m \}$ do
    if $\exists i: \bm{q}_i \circ{}\! \text{ is undefined on } a_\ell$ then
      fail;
    fi;
    differences${}\gets{}$differences${}\cup \left\{ \{ i, j \} \mid \bm{q}_i \circ a_\ell \neq \bm{q}_j \circ a_\ell \right\}$;
    for $i \in \{ 1, \dots, n + k \}$ do
      $\bm{q}_i \gets \bm{q}_i \cdot a_\ell$;
    od;
  od;
  while true do(*@ \hfill \normalfont{$\triangleright$ Guess $x \in \Sigma^*$ letter-wise until we have seen a difference for every pair @*)
    if $\forall 1 \leq i, j \leq n + k: \{ i, j \} \in{}$differences then
      return (*@\textnormal{\enquote{$u$ is $k$-expandable}}@*)(*@ \hfill \normalfont{$\triangleright$ All $\bm{q}^{(0)}_i \circ ux$ are defined and pairwise disjoint @*)
    fi;
    $b \gets{}$guess($\Sigma$);(*@ \label{ln:guessWitness} @*)(*@ \hfill \normalfont{$\triangleright$ Guess next letter of $x$ @*)
    if $\exists i: \bm{q}_i \circ{}\! \text{ is undefined on } b$ then
      fail;
    fi;
    differences${}\gets{}$differences${}\cup \left\{ \{ i, j \} \mid \bm{q}_i \circ b \neq \bm{q}_j \circ b \right\}$;
    for $i \in \{ 1, \dots, n + k \}$ do
      $\bm{q}_i \gets \bm{q}_i \cdot b$;
    od;
  od;
end;
      \end{pseudocode}
      \caption{A nondeterministic algorithm to decide whether a word is $k$-expandable.}\label{alg:expandability}
    \end{algorithm}
    \begin{cor}\label{cor:upperBoundForExpandability}
      A word $u \in \Sigma^*$ is $k$-expandable with respect to some $S$-automaton $\mathcal{T} = (Q, \Sigma, \delta)$ if and only it is already $k$-expandable by some $x \in \Sigma^*$ with
      \[
        |x| < \left( \max \{ 2, |Q| \} \right)^{|\Sigma| \, (n + k)^2} 2^{\binom{n + k}{2}}\text{,}
      \]
      where $n = |Q^* \circ u|$.
    \end{cor}
    \begin{proof}
      If $u$ is $k$-expandable, then \autoref{alg:expandability} will return \enquote{$u$ is $k$-expandable} on some computational branch. The letters guessed at \autoref{ln:guessWitness} for this branch yield a witness $x \in \Sigma^*$ for which $\left| Q^* \circ ux \right| \geq n + k$ holds. However, if, at any two points of the branch, the variables $\bm{q}_1, \dots, \bm{q}_{n + k}$ and \texttt{differences} have the same values, then the computation has a cycle and we can shorten $x$ by that cycle. This means that, without loss of generality, we can assume the length of $x$ to be bounded by the number of different configurations of these variables:
      \[
        |x| \leq \left| Q^{< K} \right|^{n + k} \cdot 2^{\binom{n + k}{2}} \text{.}
      \]
      For $|Q| \geq 2$, we have
      \[
        \left| Q^{<K} \right| = \sum_{i = 0}^{K - 1} |Q|^i = \frac{|Q|^K - 1}{|Q| - 1} < |Q|^K < |Q|^{|\Sigma| \, (n + k)}
      \]
      and, for $|Q| = 1$, we have
      \[
        \left| Q^{<K} \right| = \sum_{i = 0}^{K - 1} |Q|^i = K < 2^K < 2^{|\Sigma| \, (n + k)} \text{.}\qedhere
      \]
    \end{proof}
  \end{section}
  
  \begin{section}{Inverse Automaton Structures}\label{sec:inverse}
    \paragraph*{Expandability and Invertible Automata.}
    For inverse-deterministic au\-to\-ma\-ta, the notion of expandability is linked with an algebraic property of the stabilizer, which we will investigate next. For an $S$-automaton $\mathcal{T} = (Q, \Sigma, \delta)$, the set of \emph{stabilizing} state sequences for some word $u \in \Sigma^*$ is $\Stab_\mathcal{T}(u) = \{ \bm{q} \in Q^+ \mid \text{$\bm{q} \circ{}\!$ defined on $u$ and } q \circ u = u \}$. The subsemigroup of $\mathscr{S}(\mathcal{T})$ generated by $\Stab_\mathcal{T}(u)$ is the \emph{stabilizer} of $u$, which we denote by $\Stab_\mathcal{T}(u) \circ{}\!$. Because it will be useful, we also defined $\Stab^1_\mathcal{T}(u) = \Stab_\mathcal{T}(u) \cup \{ \varepsilon \}$.
    
    For any set of state sequences $\bm{P} \subseteq Q^*$, we also define the shifted set $\bm{P} \cdot u = \{ \bm{p} \cdot u \mid \bm{p} \in \bm{P}, \!{}\cdot u \text{ defined on } \bm{p} \}$. This, in particular, leads to the set of sequences stabilizing $u$ \emph{shifted by $u$}
    \[
      \Stab_\mathcal{T}(u) \cdot u = \{ \bm{q} \cdot u \mid \bm{q} \in \Stab_\mathcal{T}(u), \!{}\cdot u \text{ defined on $\bm{q}$} \} \text{.}
    \]
    Of course, we can also consider the orbit of some word $x \in \Sigma^*$ under the action of this set:
    \[
      \Stab^1_\mathcal{T}(u) \cdot u \circ x = \{ \bm{p} \circ x \mid \bm{p} \in \Stab^1_\mathcal{T}(u) \cdot u, \bm{p} \circ{}\! \text{ defined on } x \} \text{.}
    \]

    In this section, $\mathcal{T}$ will usually be an $\inverse{S}$-automaton. In this case, we mostly consider stabilization for the automaton $\mathcal{T} \sqcup \inverse{\mathcal{T}}$, which also includes inverse states. So, we have $\Stab^1_{\mathcal{T} \sqcup \inverse{\mathcal{T}}}(u) = \{ \bm{\wt{q}} \in \wt{Q}^* \mid \bm{\wt{q}} \circ{}\! \text{ defined on $u$ and } \bm{\wt{q}} \circ u = u \}$ for example where $\wt{Q} = Q \sqcup \inverse{Q}$ is the state set of $\mathcal{T} \sqcup \inverse{\mathcal{T}}$.
    
    With these definitions in place, we can explore the connection between an increased orbit size and the shifted stabilizer.
    \begin{lemma}\label{lem:expandedOrbitIsMultiple}
      Let $\mathcal{T} = (Q, \Sigma, \delta)$ be an $\inverse{S}$-automaton and let $\wt{Q} = Q \sqcup \inverse{Q}$ be the union of the states of $\mathcal{T}$ and $\inverse{\mathcal{T}}$.
      
      Furthermore, let $u, x \in \Sigma^*$ be words such that every $\bm{p} \circ{}\!$ for $\bm{p} \in \wt{Q}^* \cdot u$ is defined on all $y \in \Stab^1_{\mathcal{T} \sqcup \inverse{\mathcal{T}}}(u) \cdot u \circ x$. Then, we have
      \[
        \left| \wt{Q}^* \circ ux \right| = \left| \wt{Q}^* \circ u \right| \cdot \left| \Stab^1_{\mathcal{T} \sqcup \inverse{\mathcal{T}}}(u) \cdot u \circ x \right| \text{.}
      \]
      In particular, this is the case for all $u, x \in \Sigma^*$ if $\mathcal{T}$ is a $G$-automaton (i.\,e., complete).
    \end{lemma}
    \begin{proof}
      We will prove the lemma by giving a bijection
      \[
        \wt{Q}^* \circ u \times \Stab^1_{\mathcal{T} \sqcup \inverse{\mathcal{T}}}(u) \cdot u \circ x \to \wt{Q}^* \circ ux \text{.}
      \]
      For this, we choose representatives $\bm{r}_1, \dots, \bm{r}_n \in \wt{Q}^*$ for $\wt{Q}^* \circ u$ in the sense that $\{ \bm{r}_1 \circ u, \dots, \bm{r}_n \circ u \} = \wt{Q}^* \circ u$ and $n = | \wt{Q}^* \circ u |$. Similarly, we choose $\bm{s}_1, \dots, \bm{s}_m \in \Stab^1_{\mathcal{T} \sqcup \inverse{\mathcal{T}}}(u)$ with $\{ \bm{s}_1 \cdot u \circ x, \dots, \bm{s}_m \cdot u \circ x \} = \Stab^1_{\mathcal{T} \sqcup \inverse{\mathcal{T}}}(u) \cdot u \circ x$ and $m = | \Stab^1_{\mathcal{T} \sqcup \inverse{\mathcal{T}}}(u) \cdot u \circ x |$.

      Notice that, having chosen these representatives, it now suffices to give a bijection $R \times S \to \wt{Q}^* \circ ux$ for the sets $R = \{ \bm{r}_1, \dots, \bm{r}_n \}$ and $S = \{ \bm{s}_1, \dots, \bm{s}_m \}$. We define this mapping by $(\bm{r}, \bm{s}) \mapsto \bm{r} \bm{s} \circ ux = (\bm{r} \circ u) (\bm{r} \cdot u \circ (\bm{s} \cdot u \circ x))$. Notice that the right hand side is defined, in particular, $\bm{r} \cdot u \circ{}\!$ is defined on $\bm{s} \cdot u \circ x$.

      To see that the map is injective, assume $(\bm{r}, \bm{s}) \neq (\bm{r}', \bm{s}')$ for $\bm{r}, \bm{r}' \in R$ and $\bm{s}, \bm{s}' \in S$. If we have $\bm{r} \neq \bm{r}'$, then we have $\bm{r} \circ u \neq \bm{r}' \circ u$ as $\bm{r}$ and $\bm{r}'$ were chosen as representatives. If we have $\bm{r} = \bm{r}'$ but $\bm{s} \neq \bm{s}'$, then we must have $\bm{s} \cdot u \circ x \neq \bm{s}' \cdot u \circ x$ for the same reason and, since $\bm{r} \circ{}\! = \bm{r}' \circ{}\!$ and, thus, also $\bm{r} \cdot u \circ{}\! = \bm{r}' \cdot u \circ{}\!$ are one-to-one, this implies $\bm{r} \cdot u \circ \bm{s} \cdot u \circ x \neq \bm{r}' \cdot u \circ \bm{s}' \cdot u \circ x$.

      Finally, to show the surjectivity of the map, we need to find a preimage of $\bm{\wt{q}} \circ ux$ for any $\bm{\wt{q}} \in \wt{Q}^*$. Such a preimage is given by $(\bm{r}, \bm{s})$ where $\bm{r} \in R$ is chosen as the representative of $v = \bm{\wt{q}} \circ u$ (i.\,e., we have $\bm{r} \circ u = v$) and $\bm{s}$ as the representative of $x' = \inverse{\bm{r}} \bm{\wt{q}} \cdot u \circ x$ (i.\,e., we have $\bm{s} \cdot u \circ x = \inverse{\bm{r}} \bm{\wt{q}} \cdot u \circ x = x'$). For the latter, we need to observe that, under the assumptions of the lemma, $\inverse{\bm{r}} \bm{\wt{q}} \cdot u \circ{}\!$ is defined on $x$ and to show that $\inverse{\bm{r}} \bm{\wt{q}}$ is in $\Stab^1_{\mathcal{T} \sqcup \inverse{\mathcal{T}}}(u)$: we have $\inverse{\bm{r}} \bm{\wt{q}} \circ u = \inverse{\bm{r}} \circ v = \inverse{\bm{r}} \bm{r} \circ u = u$ by the choice of $\bm{r}$. To show that $(\bm{r}, \bm{s})$ is in fact a preimage of $\bm{\wt{q}} \circ uv$, it remains to observe $\bm{r} \cdot u \circ \bm{s} \cdot u \circ x = \bm{r} \cdot u \circ \inverse{\bm{r}} \bm{\wt{q}} \cdot u \circ x = \bm{r} \inverse{\bm{r}} \bm{\wt{q}} \cdot u \circ x = \bm{\wt{q}} \cdot u \circ x$.
    \end{proof}

    For the case of $G$-automata, we are able to give a more concise statement regarding the connection between expandability and the (shifted) stabilizer:
    \begin{prop}\label{prop: group case for expandability}
      Let $\mathcal{T} = (Q, \Sigma, \delta)$ be a $G$-automaton. A word $u \in \Sigma^{*} $ is expandable if and only if
      \[
        \Stab_{\mathcal{T} \sqcup \inverse{\mathcal{T}}}(u) \cdot u \circ{}\! = \{ \bm{\wt{q}} \cdot u \circ{}\! \mid \bm{\wt{q}} \in \Stab_{\mathcal{T} \sqcup \inverse{\mathcal{T}}}(u) \} \neq \{ \idGrp \}.
      \]
    \end{prop}
    \begin{proof}
      Since $\mathcal{T}$ is a $G$-automaton, we have $|Q^* \circ w| = |\wt{Q}^* \circ w|$ for every word $w \in \Sigma^*$ (see, e.\,g., \cite[Lemma~1]{decidabilityPart}). Combining this fact with \autoref{lem:expandedOrbitIsMultiple} yields
      \[
        \left| Q^* \circ ux \right| = \left| Q^* \circ u \right| \cdot \left| \Stab_{\mathcal{T} \sqcup \inverse{\mathcal{T}}}(u) \cdot u \circ x \right|
      \]
      for all words $x \in \Sigma^*$, as $\mathcal{T}$ is complete. Thus, $u$ is $k$-expandable by a word $x \in \Sigma^*$ for some integer $k \geq 1$ if and only if $\left| \Stab_{\mathcal{T} \sqcup \inverse{\mathcal{T}}}(u) \cdot u \circ x \right| > 1$, which proves the proposition.
    \end{proof}

    Notice that the analogous result for inverse automaton semigroups does not hold:
    \begin{prop}\label{prop:partialStabExpandability}
      Let $\mathcal{T} = (Q, \{ a, b \}, \delta)$ be the following $\inverse{S}$-automaton:
      
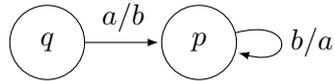
\begin{center}
        \begin{tikzpicture}[auto, shorten >=1pt, >=latex]
          \node[state] (q) {$q$};
          \node[state, right=of q] (p) {$p$};

          \path[->] (q) edge node {$a/b$} (p)
                    (p) edge[loop right] node {$b/a$} (p);
        \end{tikzpicture}
      \end{center}
      Then, $a$ is not expandable but $\Stab_{\mathcal{T} \sqcup \inverse{\mathcal{T}}}(a) \cdot a  \circ{}\!$ contains more than one element.
    \end{prop}
    \begin{proof}
      A calculation shows that $\inverse{\mathscr{S}}(\mathcal{T})$ contains the elements
      \begin{alignat*}{6}
        q \circ{}\!&:{}& ab^n &\mapsto ba^n & \qquad
        pq \circ{}\!&:{}& a &\mapsto a & \qquad
        qpq \circ{}\!&:{}& a &\mapsto b
        \\
        p \circ{}\!&:{}& b^n &\mapsto a^n &
        \inverse{q}q \circ{}\!&:{}& ab^n &\mapsto ab^n &
        pqp \circ{}\!&:{}& b &\mapsto a
        \\
        \inverse{q} \circ{}\!&:{}& ba^n &\mapsto ab^n &
        qp \circ{}\!&:{}& b &\mapsto b &
        &&&
        \\
        \inverse{p} \circ{}\!&:{}& a^n &\mapsto b^n &
        \inverse{p}p \circ{}\!&:{}& b^n &\mapsto b^n &
        &&&
        \\
        &&&&
        q\inverse{q} \circ{}\!&:{}& ba^n &\mapsto ba^n &
        &&&
        \\
        &&&&
        p\inverse{p} \circ{}\!&:{}& a^n &\mapsto a^n &
        &&&
      \end{alignat*}
      and, additionally, an element $\bot \circ{}\!$ which is undefined on any word (except the empty word) and, thus, a zero in the semigroup. The left Cayley graph and the eggbox presentation of the semigroup can be found in \autoref{fig:cayleyGraph}.
      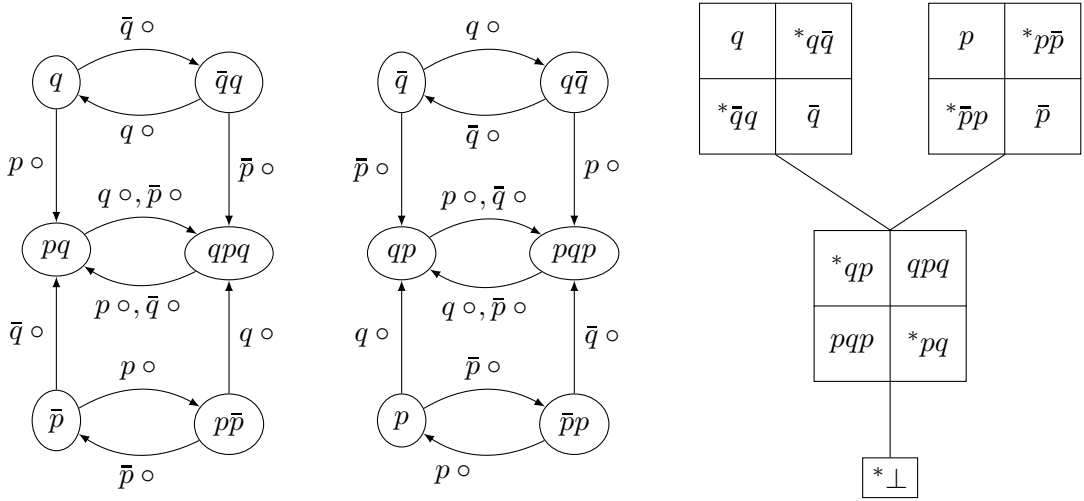
\begin{figure}\centering%
          \begin{tikzpicture}[>=latex, auto, node distance=1.5cm]
            \node (q) at (-146.32bp,94.015bp) [draw,ellipse] {$q$};
            \node[right=of q] (qiq) [draw,ellipse] {$\inverse{q}q$};
            \node[below=of q] (pq) [draw,ellipse] {$pq$};
            \node[below=of qiq] (qpq) [draw,ellipse] {$qpq$};
            \node[below=of pq] (pi) [draw,ellipse] {$\inverse{p}$};
            \node[below=of qpq] (ppi) [draw,ellipse] {$p\inverse{p}$};
            
            \node[right=of qiq] (qi) [draw,ellipse] {$\inverse{q}$};
            \node[right=of qi] (qqi) [draw,ellipse] {$q\inverse{q}$};
            \node[below=of qi] (qp) [draw,ellipse] {$qp$};
            \node[below=of qqi] (pqp) [draw,ellipse] {$pqp$};
            \node[below=of qp] (p) [draw,ellipse] {$p$};
            \node[below=of pqp] (pip) [draw,ellipse] {$\inverse{p}p$};

            \draw [->] (p) edge node {$q \circ{}\!$} (qp);
            \draw [->] (p) edge[bend left] node {$\inverse{p} \circ{}\!$} (pip);
            \draw [->] (pi) edge node {$\inverse{q} \circ{}\!$} (pq);
            \draw [->] (pi) edge[bend left] node {$p \circ{}\!$} (ppi);
            \draw [->] (pip) edge[bend left] node {$p \circ{}\!$} (p);
            \draw [->] (pip) edge node[swap] {$\inverse{q} \circ{}\!$} (pqp);
            \draw [->] (ppi) edge[bend left] node[below] {$\inverse{p} \circ{}\!$} (pi);
            \draw [->] (ppi) edge node[swap] {$q \circ{}\!$} (qpq);
            \draw [->] (pq) edge[bend left] node[above] {$q \circ{}\!, \inverse{p} \circ{}\!$} (qpq);
            \draw [->] (pqp)edge[bend left] node[below] {$q \circ{}\!, \inverse{p} \circ{}\!$} (qp);
            \draw [->] (q) edge[bend left] node {$\inverse{q} \circ{}\!$} (qiq);
            \draw [->] (q) edge node[swap] {$p \circ{}\!$} (pq);
            \draw [->] (qi) edge node[swap] {$\inverse{p} \circ{}\!$} (qp);
            \draw [->] (qi) edge[bend left] node {$q \circ{}\!$} (qqi);
            \draw [->] (qiq) edge[bend left] node[below] {$q \circ{}\!$} (q);
            \draw [->] (qiq) edge node {$\inverse{p} \circ{}\!$} (qpq);
            \draw [->] (qp) edge[bend left] node {$p \circ{}\!, \inverse{q} \circ{}\!$} (pqp);
            \draw [->] (qpq) edge[bend left] node {$p \circ{}\!, \inverse{q} \circ{}\!$} (pq);
            \draw [->] (qqi) edge[bend left] node {$\inverse{q} \circ{}\!$} (qi);
            \draw [->] (qqi) edge node {$p \circ{}\!$} (pqp);
          \end{tikzpicture}\hspace{0.75cm}
          \begin{tikzpicture}
            \matrix[matrix of math nodes, draw, every node/.style={anchor=center, minimum size=1cm}, inner sep=0pt] (Dq) {
              $q$ & ${}^* q \inverse{q}$ \\
              ${}^* \inverse{q} q$ & $\inverse{q}$ \\
            };
            \draw (Dq.north) -- (Dq.south);
            \draw (Dq.west) -- (Dq.east);
            
            \matrix[matrix of math nodes, draw, every node/.style={anchor=center, minimum size=1cm}, inner sep=0pt, right=of Dq] (Dp) {
              $p$ & ${}^* p \inverse{p}$ \\
              ${}^* \inverse{p} p$ & $\inverse{p}$ \\
            };
            \draw (Dp.north) -- (Dp.south);
            \draw (Dp.west) -- (Dp.east);

            \matrix[matrix of math nodes, draw, every node/.style={anchor=center, minimum size=1cm}, inner sep=0pt, anchor=north, yshift=-1cm] at ($(Dq.south)!0.5!(Dp.south)$) (Dqp) {
              ${}^* qp$ & $qpq$ \\
              $pqp$ & ${}^* pq$ \\
            };
            \draw (Dqp.north) -- (Dqp.south);
            \draw (Dqp.west) -- (Dqp.east);

            \node[rectangle, draw, anchor=north, below=1cm of Dqp.south] (bot) {${}^* \bot$};
            
            \draw (Dq.south) -- (Dqp.north);
            \draw (Dp.south) -- (Dqp.north);
            \draw (Dqp.south) -- (bot.north);
          \end{tikzpicture}
        \caption{The left Cayley graph and the eggbox representation of the inverse automaton semigroup generated by the automaton from \autoref{prop:partialStabExpandability}. Edges to $\bot \circ{}\!$ are not drawn.}\label{fig:cayleyGraph}
      \end{figure}

      One can observe that $\Stab_{\mathcal{T} \sqcup \inverse{\mathcal{T}}}(a) \circ{}\! = \{ pq \circ{}\!, \inverse{q}q \circ{}\!, p\inverse{p} \circ{}\! \}$ and, thus, $\Stab_{\mathcal{T} \sqcup \inverse{\mathcal{T}}}(a) \cdot a \circ{}\! = \{ pp \circ{}\! = \bot \circ{}\!, \inverse{p}p \circ{}\! \}$. Notice that, although these $pp \circ{}\! = \bot \circ{}\!$ and $\inverse{p}p \circ{}\!$ are different, there is no word $x \in \{ a, b \}^*$ such that both of them are defined on $x$ but differ in their output. This is the reason why $a$ is not expandable. To see this, consider the orbital graphs $\mathcal{T} \circ a$, $\mathcal{T} \circ aa$ and $\mathcal{T} \circ ab$ depicted in \autoref{fig:orbitalGraphs}. Notice that in the latter two, there is at most one outgoing edge at each node. Therefore, even by appending further letters to the word, the resulting orbital graph will not contain more nodes.
      \begin{figure}[t]
        \begin{center}
          \begin{tikzpicture}[>=latex, auto, shorten >=1pt,
            light/.style={
              color=gray!70, prefix after command={\pgfextra{\tikzset{every label/.style={color=gray!70}}}}},
          ]
            \node[rectangle, draw] (a) {$a$};
            \node[rectangle, draw, above=of a] (b) {$b$};
            \path[->] (a.north west) edge[bend left] node {$q$} (b.south west)
                      (b.south west) edge[bend left] node {$p$} (a.north west)
            ;

            \node[rectangle, draw, right=3cm of a] (aa) {$aa$};
            \node[rectangle, draw, above=of aa, light] (bb) {$bb$};
            \path[->, light] (bb.south west) edge[bend left] node {$p$} (aa.north west)
            ;

            \node[rectangle, draw, right=3cm of aa] (ab) {$ab$};
            \node[rectangle, draw, above=of ab] (ba) {$ba$};
            \path[->] (ab.north west) edge[bend left] node {$q$} (ba.south west)
            ;
          \end{tikzpicture}
        \end{center}
        \caption{The orbital graphs $\mathcal{T} \circ a$ (left), $\mathcal{T} \circ aa$ (center) and $\mathcal{T} \circ ab$ (right). To obtain the Schreier graphs $\wt{\mathcal{T}} \circ a$, $\wt{\mathcal{T}} \circ aa$ and $\wt{\mathcal{T}} \circ ab$, one has to add the \textcolor{gray!70}{light} node and reverse edges.}\label{fig:orbitalGraphs}
      \end{figure}
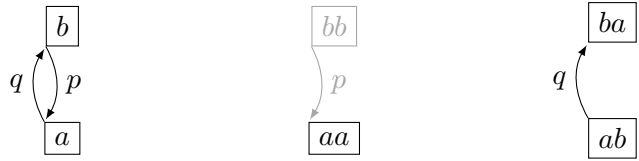
    \end{proof}

    \paragraph*{The Group Case.}
    In the case of automaton groups, it is possible to obtain a nondeterministic algorithm to decide the expandability of a given word, which is slightly more efficient than the one from \autoref{theo: decidability k-expandability}. This more efficient algorithm also yields a better upper bound on the length of a shortest suffix witnessing the expandability of a word compared to \autoref{cor:upperBoundForExpandability}. The basic idea is to exploit the equivalence between the non-triviality of the shifted stabilizer and the expandability given in \autoref{prop: group case for expandability}.
    
    \begin{lemma}\label{lem:upperBoundForShiftedStabilizerElement}
      Let $\mathcal{T} = (Q, \Sigma, \delta)$ be a $G$-automaton and let $\wt{Q} = Q \sqcup \inverse{Q}$ be the union of the states of $\mathcal{T}$ and $\inverse{\mathcal{T}}$.
      
      If the shifted stabilizer $\Stab_{\mathcal{T} \sqcup \inverse{\mathcal{T}}}(u) \cdot u \circ{}\!$ of some word $u \in \Sigma^*$ is non-trivial, then it already contains an element $\bm{\wt{q}} \in \wt{Q}^*$ with $\bm{\wt{q}} \circ{}\! \neq \idGrp$ and $|\bm{\wt{q}}| < 2 n$ where $n = |Q^* \circ u|$.
    \end{lemma}
    \begin{proof}
      Consider the Schreier graph $\wt{\mathcal{T}} \circ u$. It contains $|\wt{\mathcal{T}} \circ u| = |\wt{Q}^* \circ u| = n$ elements (since Schreier and orbital graphs are of the same size, see, e.\,g., \cite[Lemma~1]{decidabilityPart}). Obviously, the elements $\bm{\wt{q}} \in \Stab^1_{\mathcal{T} \sqcup \inverse{\mathcal{T}}}(u)$ of the stabilizer are in one-to-one correspondence with the labels of loops beginning and ending in $u$. If we fix some spanning tree for the Schreier graph $\wt{\mathcal{T}} \circ u$, then every edge $v_1 \overset{q}{\rightarrow} v_2$ \emph{not} belonging to this spanning tree induces a loop $c$ beginning and ending in $u$: first, follow the unique path from $u$ to $v_1$ on the spanning tree, then, take the edge $v_1 \overset{q}{\rightarrow} v_2$ itself, finally, return from $v_2$ to $u$ on the spanning tree again. Notice that this loop contains at most $n - 1 + 1 + n - 1 = 2n - 1$ edges since any (reduced) path on the spanning tree can visit any node at most once. There are only finitely many such loops and they generate the set of all loops beginning and ending in $u$. Let $C = \{ c_1, \dots, c_m \}$ denote the set of the labels $c_i \in \wt{Q}^*$ of these generating loops. Then, $C \circ{}\! = \{ c_1 \circ{}\!, \dots, c_m \circ{}\! \}$ generates the stabilizer $\Stab^1_{\mathcal{T} \sqcup \inverse{\mathcal{T}}}(u) \circ{}\!$ and its shifted version $C \cdot u \circ{}\! = \{ c_1 \cdot u \circ{}\!, \dots, c_m \cdot u \circ{}\! \}$ is a generating set for the shifted stabilizer $\Stab_{\mathcal{T} \sqcup \inverse{\mathcal{T}}}(u) \cdot u \circ{}\!$. Since, by assumption, this shifted stabilizer contains another element besides the identity, there must be at least one generator $c_i \cdot u \circ{}\!$ with $c_i \cdot u \circ{}\! \neq \idGrp$. This is the sought element.
    \end{proof}
    
    This allows us to give a guess and check algorithm to decide the expandability of a given word, which is simpler and more efficient than \autoref{alg:expandability}.
    \begin{thm}
      The problem
      \problem{
        a $G$-automaton $\mathcal{T} = (Q, \Sigma, \delta)$ and\newline
        a word $u \in \Sigma^*$
      }{
        is $u$ expandable (with respect to $\mathcal{T}$)?
      }
      \noindent{}can be decided in nondeterministic space $\mathcal{O}(n(1 + \log |Q|) + |u|(1 + \log |\Sigma|))$ where $n = |Q^* \circ u| \leq |\Sigma|^{|u|}$. For the problem
      \problem[
        a $G$-automaton $\mathcal{T} = (Q, \Sigma, \delta)$
      ]{
        a word $u \in \Sigma^*$
      }{
        is $u$ expandable (with respect to $\mathcal{T}$)?
      }
      \noindent{}this yields nondeterministic space $2^{\mathcal{O}(|u|)}$.
    \end{thm}
    \begin{proof}
      As in the proof of \autoref{theo: decidability k-expandability}, we first compute $n$ using the technique of inductive counting. This is possible in space $\mathcal{O}(n(1 + \log |Q|) + |u|(1 + \log |\Sigma|))$.
      
      Then, we guess a sequence of states $\bm{\wt{q}} \in \wt{Q}^*$ (where $\wt{Q} = Q \sqcup \inverse{Q}$ is the union of the states of $\mathcal{T}$ and its inverse $\inverse{\mathcal{T}}$) with length $|\bm{\wt{q}}| < 2n$. To store such a sequence, we need space $\mathcal{O}(n(1 + \log |Q|))$. Next, we compute $\bm{\wt{q}} \cdot u$ letter by letter (of $u$). Simultaneously, we check whether $\bm{\wt{q}} \in \Stab^1_{\mathcal{T} \sqcup \inverse{\mathcal{T}}}(u)$ by computing $\bm{\wt{q}} \circ u$ and comparing the result with $u$ (again, we do this letter by letter). For this, we need to store only single letters ($\mathcal{O}(1 + \log |\Sigma|)$) and some pointer ($\mathcal{O}(\log |u|)$).
      
      Finally, we solve the word problem \enquote{$\bm{\wt{q}} \cdot u \circ{}\! \neq \idGrp$?} by guessing a witness $x \in \Sigma^*$ with $\bm{\wt{q}} \cdot u \circ x \neq x$. As before, we guess $x$ letter by letter and update the stored state sequence accordingly. At the same time, we check whether at least one output letter differs from its input letter.
      
      Clearly, if we can guess a state sequence $\bm{\wt{q}}$ and a witness $x \in \Sigma^*$ with $\bm{\wt{q}} \circ u = u$, $\bm{\wt{q}} \cdot u \circ x \neq x$, the shifted stabilizer of $u$ is non-trivial and $u$ is expandable (by \autoref{prop: group case for expandability}). If, on the other hand, $u$ is expandable, then its shifted stabilizer is non-trivial (again, by \autoref{prop: group case for expandability}) and, by \autoref{lem:upperBoundForShiftedStabilizerElement}, it contains a non-trivial element of length smaller than $2n$. Some computational branch of the above algorithm will guess this element $\bm{\wt{q}}$ and a corresponding witness $x$ for its non-triviality.
    \end{proof}
    
    Analyzing the last part of the presented algorithm yields a better upper bound for the witness $x$ in the group case (compared to the general case presented in \autoref{cor:upperBoundForExpandability}).
    \begin{cor}\label{cor:upperBoundForExpandabilityGroup}
      A word $u \in \Sigma^*$ is expandable with respect to some $G$-automaton $\mathcal{T} = (Q, \Sigma, \delta)$ if and only if it is already expandable by some $x \in \Sigma^*$ with
      \[
        |x| < (2 \, |Q|)^{2n}
      \]
      where $n = |Q^* \circ u|$.
    \end{cor}
    \begin{proof}
      We use the same analysis as in the proof of \autoref{cor:upperBoundForExpandability}: consider the last part of the algorithm, in which the witness $x$ is guessed letter by letter. If, during this guessing, the stored state sequence $(\bm{\wt{q}} \cdot u) \cdot x$ has the same value as some time before on the same computational branch, then this computational loop can be eliminated. Thus, we only need to count the possible values for the stored state sequence:
      \[
        \left| \wt{Q}^{<2n} \right| < |\wt{Q}|^{2n} = (2 \, |Q|)^{2n} \qedhere
      \]
    \end{proof}
  \end{section}

  \begin{section}{Reversible and Complete Automata.}\label{sec:reversible}
    We have seen that expandability is decidable. Next, we will focus on the special case of reversible and complete automata.

    We start, however, with the following lemma, which does not require reversibility or completeness.\footnote{The lemma is a rather simple generalization of \cite[Corollary~3.5]{cain2009automaton}. It is closely related to the word problem for automaton structures, for which a more thorough discussion -- including some lower bound results -- can be found in \cite{DAngeli2017} (see also \cite{pspacePart}).}
    \begin{lemma}\label{lem:wordProblemUpperBound}
      Let $\mathcal{T} = (Q, \Sigma, \delta)$ be an $S$-automaton and let $\bm{p}, \bm{q}, \bm{r} \in Q^*$. If there is a word on which $\bm{p} \bm{r}$ and $\bm{q} \bm{r}$ act differently, then there is such a word of length at most $|Q|^{|\bm{p}| + |\bm{q}| + |\bm{r}|}$.
    \end{lemma}
    \begin{proof}
      The upper bound can be obtained in a similar way to \autoref{cor:upperBoundForExpandability} and \autoref{cor:upperBoundForExpandabilityGroup}: we give a non-deterministic algorithm for the problem
      \problem
        [an $S$-automaton $\mathcal{T} = (Q, \Sigma, \delta)$]
        {$\bm{p}, \bm{q}, \bm{r} \in Q^*$}
        {$\exists u \in \Sigma^*: \bm{p} \bm{r} \circ u \neq \bm{q} \bm{r} \circ u$?}
      \noindent{}and analyze the number of possible configurations. The algorithm is a quite straight-forward guess and check approach. We need three variables holding values from $Q^+$, which correspond to the three input values. We guess a witness $u \in \Sigma^*$ letter by letter. If we guess $a \in \Sigma$, we compute $b = \bm{r} \circ a$ and (simultaneously) update $\bm{r} \gets \bm{r} \cdot a$. Then, we check whether we have $\bm{p} \circ b \neq \bm{q} \circ b$. If this is true, then we have found a witness. Simultaneously to the check, we update $\bm{p} \gets \bm{p} \cdot b$ and $\bm{q} \gets \bm{q} \cdot b$.

      If, during the computation of a witness, all three variables $\bm{r}$, $\bm{p}$ and $\bm{q}$ have the same values as sometime before on the same computational branch, we can eliminate this computational loop and get a shorter witness. Notice that, during any computation, the variable $\bm{r}$ will always hold a value of the same length as its initial value and that the same is true for $\bm{p}$ and $\bm{q}$. The upper bound for a shortest witness stated in the lemma follows by calculating the number of different possible values for the triple of variables $(\bm{r}, \bm{p},\bm{q})$.
    \end{proof}

    The next straightforward lemma starts our discussion of reversible automata.
    \begin{lemma}\label{lem: reversibility on defining relations}
      Let $\mathcal{T} = (Q, \Sigma, \delta)$ be a complete and reversible $S$-automaton. Then, for every $\bm{p}, \bm{q} \in Q^+$ and every $u \in \Sigma^*$, we have $\bm{p} = \bm{q}$ in $\mathscr{S}(\mathcal{T})$ if and only if $\bm{p} \cdot u = \bm{q} \cdot u$ in $\mathscr{S}(\mathcal{T})$.
    \end{lemma}
    \begin{proof}
      If $\bm{p} = \bm{q}$ in $\mathscr{S}(\mathcal{T})$, then it is clear that $\bm{p} \cdot u = \bm{q} \cdot u$ in $\mathscr{S}(\mathcal{T})$ also holds.

      For the converse, recall that, for every word $u \in \Sigma^*$, the mapping $\!{}\cdot u \colon Q^{+} \to Q^{+}, \bm{q} \mapsto \bm{q} \cdot u$ is a total, length-preserving bijection (by \autoref{fct:reversibleImpliesOneToOne}). Let $\bm{p}, \bm{q} \in Q^+$ such that $\bm{p} \cdot u = \bm{q} \cdot u$ in $\mathscr{S}(\mathcal{T})$. Because $\!{}\cdot u$ is length-preserving and bijective, there are $k, \ell > 0$ such that $\bm{p} \cdot u^k = \bm{p}$ and $\bm{q} \cdot u^\ell = \bm{q}$. Taking $m$ as the least common multiple of $k$ and $\ell$, we obtain
      \[
        \bm{p} = \left( \bm{p} \cdot u \right) \cdot u^{m - 1} \text{ and } \bm{q} = \left( \bm{q} \cdot u \right) \cdot u^{m - 1} \text{,}
      \]
      which implies that $\bm{p}$ and $\bm{q}$ are equal in $\mathscr{S}(\mathcal{T})$ by the same argument as for the direct direction.
    \end{proof}

    The observation from the previous lemma is sufficient to show that, in an infinite automaton semigroup generated by a complete and reversible automaton, every word is expandable.

    \begin{thm}\label{thm:upperBoundForExpandabilityReversibleComplete}
      Let $\mathcal{T} = (Q, \Sigma, \delta)$ be a complete and reversible $S$-automaton such that $\mathscr{S}(\mathcal{T})$ is infinite. Then, every word $u \in \Sigma^*$ is expandable by a word of length at most $|Q|^n$ where $n = \left| Q^* \circ u \right|$.
    \end{thm}
    \begin{proof}
      By the pigeon hole principle, there must be two state sequences $\bm{p}, \bm{q} \in Q^+$ such that $\bm{p} \circ u = \bm{q} \circ u$ but $\bm{p} \neq \bm{q}$ in $\mathscr{S}(\mathcal{T})$ since $\mathscr{S}(\mathcal{T})$ contains infinitely many functions while the possible images of $u$ are a subset of $\Sigma^{|u|}$ and, thus, finite. By \autoref{lem: reversibility on defining relations}, this implies $\bm{p} \cdot u \neq \bm{q} \cdot u$ in $\mathscr{S}(\mathcal{T})$. Thus, there is a word $x \in \Sigma^*$ such that $\bm{p} \cdot u \circ x \neq \bm{q} \cdot u \circ x$. This means $\bm{p} \circ ux \neq \bm{q} \circ ux$. As $\mathcal{T}$ is complete, the mapping $Q^* \circ ux \to Q^* \circ u, \bm{q} \circ ux \mapsto \bm{q} \circ u$ is surjective. The two distinct elements $\bm{p} \circ ux$ and $\bm{q} \circ ux$ map to the same element $\bm{p} \circ u = \bm{q} \circ u$, however; thus, it is not injective. Therefore, we have $\left| Q^* \circ ux \right| > \left| Q^* \circ u \right|$ or, in other words, that $u$ is expandable (by $x$).

      To prove the stated upper bound on the length, we consider $\bm{p}$ and $\bm{q}$ as labeled paths in the orbital graph\footnote{Remember that, because $\circ$ is a \emph{left} action, the last \emph{letter} of $\bm{q}$ is the label of the \emph{first} edge.} $\mathcal{T} \circ u$. Suppose that $\bm{\ell}$ is a cycle on the path $\bm{q} = \bm{q}_2 \bm{\ell} \bm{q}_1$:
      
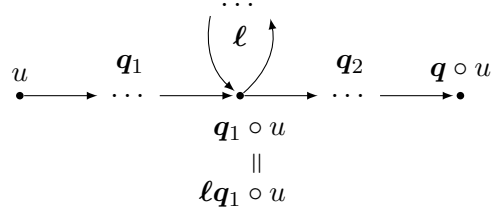
\begin{center}
        \begin{tikzpicture}[auto, shorten >=1pt, >=latex,
          vertex/.style={circle, fill=black, draw=none, inner sep=1pt}]

          \node[vertex, label=above:$u$] (u) {};
          \node[right=of u] (dl) {$\dots$};
          \node[vertex, right=of dl, label={[align=right]below:$\begin{aligned}
              \bm{q}_1 &\circ u\\[-0.25\baselineskip]
              &{}\mathrel{\rotatebox[origin=c]{-90}{{}={}}}{}\\[-0.25\baselineskip]
              \bm{\ell} \bm{q}_1 &\circ u
            \end{aligned}$}] (c) {};
          \node[right=of c] (dr) {$\dots$};
          \node[vertex, right=of dr, label=above:$\bm{q} \circ u$] (qu) {};
          \node[above=of c] (dt) {$\dots$};

          \node[above=0cm of dl] {$\bm{q}_1$};
          \node[below=0cm of dt] {$\bm{\ell}$};
          \node[above=0cm of dr] {$\bm{q}_2$};

          \path[->] (u) edge (dl)
                    (dl) edge (c)
                    (c) edge (dr)
                    (dr) edge (qu)
                    (c) edge[bend right] (dt.south east)
                    (dt.south west) edge[bend right] (c);
        \end{tikzpicture}
      \end{center}
      First, suppose that we have $\bm{q}_2 \bm{q}_1 \neq \bm{q} = \bm{q}_2 \bm{\ell} \bm{q}_1$ in $\mathscr{S}(\mathcal{T})$. This is only possible if $\bm{q}_1 \neq \bm{\ell} \bm{q}_1$ in $\mathscr{S}(\mathcal{T})$. Since we still have $\bm{q}_1 \circ u = \bm{\ell \bm{q}}_1 \circ u$, we can replace $\bm{q}$ by $\bm{q}_1$ and $\bm{p}$ by $\bm{\ell} \bm{q}_1$. We may assume $\bm{\ell}$ to be the first cycle on the path and to be of minimal length. Thus, the only node shared between the two paths is $\bm{q}_1 \circ u = \bm{\ell \bm{q}}_1 \circ u$, which yields $|\bm{q}_1| + 1 + |\bm{\ell}| - 1 \leq n$. Using \autoref{lem:wordProblemUpperBound}, this yields the upper bound of $|Q|^n$ on the length of the witness $x$ since $\bm{q}$ and $\bm{q} \cdot u$ are of the same length.

      If we have $\bm{q}_2 \bm{q}_1 = \bm{q}$ in $\mathscr{S}(\mathcal{T})$, then we can remove the cycle without changing the semigroup element (i.\,e., we still have $\bm{q}_2 \bm{q}_1 \neq \bm{p}$ in $\mathscr{S}(\mathcal{T})$ but $\bm{q}_2 \bm{q}_1 \circ u = \bm{p} \circ u$). Iterating this argument, we may assume $\bm{p}$ and $\bm{q}$ not to contain cycles. They may still have nodes in common, however. Let $\bm{p_1}$ be their longest common prefix (seen as paths), i.\,e., they part afterwards and join again later (because they both end in the same node):
      
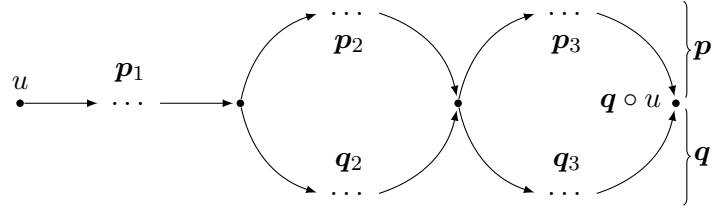
\begin{center}
        \begin{tikzpicture}[auto, shorten >=1pt, >=latex,
          vertex/.style={circle, fill=black, draw=none, inner sep=1pt}]

          \node[vertex, label=above:$u$] (u) {};
          \node[right=of u] (dl) {$\dots$};
          \node[vertex, right=of dl] (s) {};
          \node[above right=of s] (dct) {$\dots$};
          \node[below right=of s] (dcb) {$\dots$};
          \node[vertex, below right=of dct] (j) {};
          \node[above right=of j] (drt) {$\dots$};
          \node[below right=of j] (drb) {$\dots$};
          \node[vertex, below right=of drt, label=left:$\bm{q} \circ u$] (qu) {};

          \node[above=0cm of dl] {$\bm{p}_1$};
          \node[below=0cm of dct] {$\bm{p}_2$};
          \node[above=0cm of dcb] {$\bm{q}_2$};
          \node[below=0cm of drt] {$\bm{p}_3$};
          \node[above=0cm of drb] {$\bm{q}_3$};

          \path[->] (u) edge (dl)
                    (dl) edge (s)

                    (s) edge[bend left] (dct.west)
                    (dct.east) edge[bend left] (j)
                    (s) edge[bend right] (dcb.west)
                    (dcb.east) edge[bend right] (j)

                    (j) edge[bend left] (drt.west)
                    (drt.east) edge[bend left] (qu)
                    (j) edge[bend right] (drb.west)
                    (drb.east) edge[bend right] (qu);

          \draw[decorate,decoration={brace, raise=1mm}]
            (drt.north east -| qu.north) -- node[xshift=1mm] {$\bm{p}$} (qu.north);
          \draw[decorate,decoration={brace, mirror, raise=1mm}]
            (drb.south east -| qu.south) -- node[right, xshift=1mm] {$\bm{q}$} (qu.south);
        \end{tikzpicture}
      \end{center}
      By choosing the paths $\bm{p}_2$ and $\bm{q}_2$ of minimal length, we may assume that they do not have any common nodes (except for their first and last node, of course). The first case is that $\bm{p}_2 \bm{p}_1 \neq \bm{q}_2 \bm{p}_1$ in $\mathscr{S}(\mathcal{T})$. As we still have $\bm{p}_2 \bm{p}_1 \circ u = \bm{q}_2 \bm{p}_1 \circ u$, we can replace $\bm{p}$ by $\bm{p}_2 \bm{p}_1$ and $\bm{q}$ by $\bm{q}_2 \bm{p}_1$. Notice that, on both paths, there are $|\bm{p}_1| + 1 + |\bm{p}_2| + |\bm{q}_2| - 1 \leq n$ distinct nodes. Thus, \autoref{lem:wordProblemUpperBound} yields a witness $x$ of length at most $|Q|^n$. If, on the other hand, we have $\bm{p}_2 \bm{p}_1 = \bm{q}_2 \bm{p}_1$ in $\mathscr{S}(\mathcal{T})$, we can replace $\bm{q}$ by $\bm{q}_3 \bm{p}_2 \bm{p}_1$ and iterate the argument.
    \end{proof}
  \end{section}
  
  \begin{section}{Lower Bounds}\label{sec:lowerBound}
      Besides the relatively large upper bound of about $|Q|^{n^2}$ in the semigroup case for the length of an expanding word given in \autoref{cor:upperBoundForExpandability} and the only slightly lower upper bounds of $(2 \, |Q|)^{2n}$ and $|Q|^n$ given in \autoref{cor:upperBoundForExpandabilityGroup} for the group case and in \autoref{thm:upperBoundForExpandabilityReversibleComplete} for the complete reversible case, it seems to be rather difficult to obtain constructions for lower bounds. The following rather straightforward modification of the adding machine yields a constant lower bound (where the constant is approximately the number of states).
      \begin{prop}\label{prop:lowerBoundForExpandabilityGroupCase}
        Let $\ell > 0$ and $\mathcal{T} = (Q, \Sigma, \delta)$ be the $G$-automaton
        
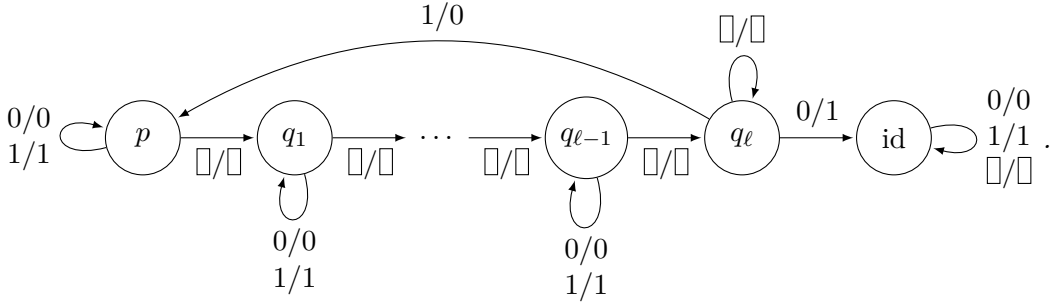
\begin{center}
          \begin{tikzpicture}[auto, shorten >=1pt, >=latex, baseline=(id.base)]
            \node[state] (p) {$p$};
            \node[state, right=of p] (q1) {$q_1$};
            \node[right=of q1] (dots) {$\dots$};
            \node[state, right=of dots] (ql-1) {$q_{\ell - 1}$};
            \node[state, right=of ql-1] (ql) {$q_\ell$};
            \node[state, right=of ql] (id) {$\id$};
            
            \draw[->] (p) edge node[swap] {$\blank / \blank$} (q1)
                      (q1) edge node[swap] {$\blank / \blank$} (dots)
                      (dots) edge node[swap] {$\blank / \blank$} (ql-1)
                      (ql-1) edge node[swap] {$\blank / \blank$} (ql)
                      (ql) edge[bend right] node[swap] {$1 / 0$} (p)
                      (ql) edge node {$0 / 1$} (id)
                      (id) edge[loop right] node[align=center] {$0/0$\\$1/1$\\$\blank / \blank$} (id)
            ;
            \draw[->] (p) edge[loop left] node[align=center] {$0/0$\\$1/1$} (p)
                      (q1) edge[loop below] node[align=center] {$0/0$\\$1/1$} (q1)
                      (ql-1) edge[loop below] node[align=center] {$0/0$\\$1/1$} (ql-1)
                      (ql) edge[loop above] node {$\blank / \blank$} (ql)
            ;
          \end{tikzpicture}.
        \end{center}
        Then, for every $n \geq 1$, the word $\left( \blank^\ell 0 \right)^n$ is only expandable by words $x \in \Sigma^*$ of length $|x| \geq \ell + 1 = |Q| - 1$.
      \end{prop}
      \begin{proof}
        First, note that the shortest words on which the actions of $p$ and $\id$ differ are $\blank^\ell 0$ and $\blank^\ell 1$, which are of length $\ell + 1$. As a consequence, we have $\bm{q} \circ u = \bm{p} \circ u$ for all $\bm{q}, \bm{p} \in \{ p, \id \}^*$ if $|u| < \ell + 1$.

        Next, note that, for all $q \in Q$, we have $q \cdot \blank^\ell z \in \{ p, \id \}$ for $z \in \{ 0, 1 \}$ and that the orbit of $\left( \blank^\ell 0 \right)^n$ is $Q^* \circ \left( \blank^\ell 0 \right)^n = \left( \blank^\ell \{ 0, 1 \} \right)^n$. Also, note that, in fact, $Q^* \cdot \left( \blank^\ell 0 \right)^n = \{ p, \id \}^*$, which concludes the argument.
      \end{proof}

      Obviously, there is a huge gap between the above lower bound construction and the upper bounds, which leads to the following open problem.
      \begin{prob}
        Let $\mathcal{T} = (Q, \Sigma, \delta)$ be an $S$-/$\inverse{S}$-/$G$-automaton such that $u \in \Sigma^*$ is expandable and let $x$ be a shortest word $x \in \Sigma^*$ for which $|Q^* \circ ux|$ is larger than $|Q^* \circ u|$. Are there better upper bounds for $|x|$ than the ones stated in \autoref{cor:upperBoundForExpandability}, \autoref{cor:upperBoundForExpandabilityGroup} and \autoref{thm:upperBoundForExpandabilityReversibleComplete}?

        On the other hand, are there better lower bounds $b(n)$ than the one given by \autoref{prop:lowerBoundForExpandabilityGroupCase} such that there are $S$-/$\inverse{S}$-/$G$-automata $\mathcal{T}_n = (Q_n, \Sigma_n, \delta_n)$ and words $u_n \in \Sigma_n^n$ such that $u_n$ is not expandable by any word of length smaller than $b(n)$?
      \end{prob}
    \end{section}
    
\newpage
\bibliographystyle{plain}
\bibliography{references}

\end{document}